\documentclass[11pt]{article}

\usepackage{graphicx}
\usepackage{latexsym}
\usepackage{fancyhdr}
\usepackage{subfigure}
\usepackage{mdwlist}
\usepackage{amssymb}

\makeatletter

\newenvironment{proof}{\paragraph{Proof:}}{\hspace*{\fill}\(\Box\)}

\newtheorem{theorem}{Theorem}

\newtheorem{lemma}{Lemma}

\def\noflash#1{\setbox0=\hbox{#1}\hbox to 1\wd0{\hfill}}

\newcommand{\comment}[1]{}

\newcommand{\nocomment}[1]{}

\newcommand{\Iitemize}{\begin{itemize}
	{\setlength{\itemsep}{-6pt}}
       }

\newcommand{\ls}[1]
   {\dimen0=\fontdimen6\the\font 
    \lineskip=#1\dimen0
    \advance\lineskip.5\fontdimen5\the\font
    \advance\lineskip-\dimen0
    \lineskiplimit=.9\lineskip
    \baselineskip=\lineskip
    \advance\baselineskip\dimen0
    \normallineskip\lineskip
    \normallineskiplimit\lineskiplimit
    \normalbaselineskip\baselineskip
    \ignorespaces
   }

\def\ifundefined#1{\expandafter\ifx\csname#1\endcsname\relax}
\ifundefined{newblock}{
 }\fi

\newcommand{\eqref}[1]{Equation~\ref{#1}}

\begin{document}

\title{Multiparty Equality Function Computation in\\ Networks with Point-to-Point Links\footnote{\normalsize This research is supported
in part by Army Research Office grant
W-911-NF-0710287 and National Science Foundation award 1059540. Any opinions, findings, and conclusions or recommendations expressed here are those of the authors and do not
necessarily reflect the views of the funding agencies or the U.S. 
government.}}

\date{\today}
\author{Guanfeng Liang and Nitin Vaidya\\ \normalsize Department of Electrical and Computer Engineering, and\\ \normalsize Coordinated Science Laboratory\\ \normalsize University of Illinois at Urbana-Champaign\\ \normalsize gliang2@illinois.edu, nhv@illinois.edu\\~\\Technical Report}

%


\maketitle


\thispagestyle{empty}

\newpage

\setcounter{page}{1}

\section{Introduction}\label{sec:intro}
In this report, we study the multiparty communication complexity problem of the multiparty equality function (MEQ):
\begin{equation}
EQ(x_1,\cdots,x_n) = \left\{ 
\begin{array}{ll}
0 & \textrm{if } x_1=\cdots=x_n\\
1 & \textrm{otherwise}.
\end{array}
\right.
\label{eq:EQ}
\end{equation}
 The input vector $x=(x_1,\cdots,x_n)$ is distributed among $n\ge 2$ nodes, with $x_i$ known to node $i$, where $x_i$ is chosen from the set $\{1,\cdots,M\}$, for some integer $M>0$.

\subsection{Communication Complexity}\label{subsec:cc}
The notion of communication complexity (CC) was introduced by Yao in 1979 \cite{Yao'79}, who investigated the following problem involving two separated parties (Alice and Bob) want to mutually compute a Boolean function that is defined on pairs of inputs. Formally, let $f:X\times Y \mapsto \{0,1\}$ be a Boolean function. The communication problem for $f$ is the following two-party game:

Alice receives $x\in X$ and Bob receives $y\in Y$, and the goal is for them to compute $f(x,y)$, collaboratively. Alice and Bob have unlimited computational power and a full description of $f$, but they do not know each other's input. They determine the output value by exchanging messages. The computation ends when either Alice or Bob has enough information to determine $f(x,y)$, and sends a special symbol ``halt'' to the other party.

A protocol $P$ for computing $f$ is an algorithm, according to which Alice and Bob send binary messages to each other. A protocol proceeds in rounds. In every round, the protocol specifies whose turn it is to send a message. Each party in his/her turn sends one bit that may depend on his/her input and the previous messages he/she has received. A correct protocol for $f$ should terminate for every input pair $(x,y)\in X\times Y$, when either Alice or Bob knows $f(x,y)$.

The communication complexity of a protocol $P$ is the number of bits exchanged for the worst case input pair. The communication complexity of a Boolean function $f:X\times Y \mapsto \{0,1\}$, is that of the protocols for $f$ with the least complexity.

\subsection{Multiparty Communication Complexity}\label{subsec:Mcc}
There is more than one way to generalize communication complexity to a multiparty setting. The most commonly used model is the ``number on the forehead'' model introduced in \cite{multi-partyprotocols}. Formally, there is some function $f : \Pi _{i=1}^n X_i \mapsto \{0, 1\}$, and the input
is $(x_1, x_2,\cdots, x_n)$ where each $x_i \in X_i$. The $i$-th party can see all the $x_j$ such that $j \neq i$. As
in the 2-party case, the $n$ players have an agreed-upon protocol for communication, and all this
communication is posted on a ``public blackboard''. At the end of the protocol all parties must know $f(x_1,\cdots,x_n)$. In this model, the communications may be considered as being {\em broadcast} using the public blackboard. Tight lower bounds (at least up to constant factors) often follow from considering two-way partitions of the set of parties.

\section{Models}\label{sec:model}

\subsection{Communication Model}
Instead of the ``number on the forehand'' model, we consider a point-to-point  communication model (similar to the message passing model), which we believe is more realistic in networking settings.
We assume a synchronous fully connected network of $n$ nodes, the node IDs (identifiers) are common knowledge. We assume that all point-to-point communication channels/links are {\em private} such that when a node transmits, only the designated recipient can receive the message. The identity of the sender is known to the recipient.

\subsection{Protocol}
A protocol $P$ is a sequence/schedule of transmissions and computations $\{\alpha_1\alpha_2\cdots\alpha_{L(P)}\}$. Here $\alpha_l = (T_l,R_l,f_l(x_{T_l}, T_l^+(l)))$ specifies that in the $l$-th step, node $T_l$ sends a channel symbol as a function $f_l(x_{T_l},T_l^+(l))$ to node $R_l$, with $T_l^+(l)$ denoting all the channel symbols party $T_l$ has received up to step $l-1$. 
In the rest of the paper, we will use $T_l^+$ as $T_l^+(l)$ when it is clear from the context. $L(P)$ is  the length of protocol $P$. The communication complexity of protocol $P$ is defined as
\begin{equation}
C(P) = \sum_{l=1}^{L(P)}\log_2 S_l(P),
\end{equation}
where $S_l(P)$ is the size of the {\em range} of $f_l(x_{T_l},T_l^+(l))$, i.e., the number of possible channel symbols needed in step $l$ of protocol $P$, considering all possible inputs. If only binary symbols are allowed, the communication complexity simply becomes $L(P)$.

\section{Problem Definitions}\label{sec:problem}
In this paper, we study the communication complexity of computing the MEQ function in a distributed manner. We consider two versions of the problem, which differ in where the MEQ function is being computed.

\subsection{MEQ-AD: Anyone Detects}
We first consider the protocols that terminates whenever one of the nodes detects a {\em mismatch}. Precisely speaking, a protocol $P$ is said to solve the MEQ-AD problem {\em deterministically} if by the end of the $L(P)$-th step, every node computes its own decision $EQ_i\in\{0,1\}$ such that
\begin{equation}
EQ_1=\cdots EQ_n=0 \Leftrightarrow EQ(x_1,\cdots,x_n)=0.
\end{equation} 
We will say node $i$ detects a mismatch if it sets $EQ_i=1$. In other words, if the inputs to the $n$ parties are not identical, there must be {\bf at least one} node that detects a mismatch.

\subsection{MEQ-CD: Centralized Detect}
The second class of protocols we consider are the ones in which a particular node is assigned to perform mismatch detection. Without loss of generality, we can assume that node $n$ has to perform detection. Then a protocol $P$ is said to solve the MEQ-AD problem if by the end of the $L(P)$-th step, node $n$ computes decision $EQ_n$, which is 
\begin{equation}
EQ_n= EQ(x_1,\cdots,x_n).
\end{equation} 

\subsection{Goal}
Denote $\Gamma_{AD}(n,M)$ and $\Gamma_{CD}(n,M)$ as the set of all protocols that solve the MEQ-AD and MEQ-CD problem with $n$ nodes, each of which is given an input value $x_i\in\{1,\cdots,M\}$, respectively. We are interested in finding the  communication complexity of the both problems, which are defined as
\begin{eqnarray}
C_{AD}(n,M) = \inf_{P\in \Gamma_{AD}(n,M)}C(P),\\
C_{CD}(n,M) = \inf_{P\in \Gamma_{CD}(n,M)}C(P).
\end{eqnarray}
It is worth pointing out that $C_{AD}(n,M)\le C_{CD}(n,M)$, since any protocol that solves the MEQ-CD problem solves the MEQ-AD problem with the same $n$ and $M$ as well.
In the rest of this report, we will mainly focus on the MEQ-AD problem, except for Section \ref{sec:MEQ-CD}, in which we discuss the MEQ-CD problem.

\section{Upper Bound of the Complexity}\label{sec:upper}
An upper bound of the communication complexity of both versions of the MEQ problem is $ (n-1)\log_2 M$, for all positive integer $n\ge 2$ and $M\ge 1$.
This can be proved by a trivial construction: in step $i$, node $i$ sends $x_{i}$ to node $n$, for all $i<n$. The decisions are computed according to
\begin{equation}
EQ_i = \left\{ 
\begin{array}{ll}
EQ(x_1,\cdots,x_n) & , i=n;\\
0 &  ,i<n.
\end{array}
\right.
\end{equation}
It is obvious that this protocol solves both the MEQ-AD and MEQ-CD problems with communication complexity $(n-1)\log_2 M$, which implies $C_{AD(CD)}(n,M)\le (n-1)\log_2 M$. In particular, when $M=2^k$, we have $C_{AD(CD)}(n,2^k)\le (n-1)k$. 


\section{Loose Lower Bound of Complexity using Traditional Techniques}\label{sec:lower}
In most of the existing literature on multiparty communication complexity, the ``number on the forehand'' model or a broadcast communication model is usually assumed. Under these models, when a node transmits, all other nodes receive the same message. This broadcasting property makes it possible to consider two-way partitions of the set of nodes since the nodes in each partition shares the same information being broadcast and can be viewed as one virtual node. Thus results from two-party communication complexity can be extended to the multiparty case, and tight bounds (rather than just capturing the order) can then be obtained. 

However, the above technique no longer works well in obtaining tight bounds under our point-to-point communication model. For example, the complexity of the two-party EQ problem of $k$-bit inputs can be proved to be $k$ with the ``fooling set'' argument: Suppose in contradiction that there exists a protocol of  complexity at most $C(P)<k$ that solves the two-party EQ problem. Then there are at most $2^{C(P)} \le 2^k-1$ communication patterns possible between the two nodes. Consider all sets of $2^k$ pairs of input values $(x,x)$. Using the pigeonhole principle we conclude there exist two
pairs $(x, x)$ and $(x', x')$ on which the communication patterns are the same. It is easy to see that the communication pattern of $(x, x')$ is also the same as $(x,x)$. Hence, the nodes' final decisions on $(x,x)$ must agree with their decisions on $(x, x')$. But then the protocol must be incorrect, since $EQ(x,x')=1\neq EQ(x,x)$.

The ``fooling set'' argument above can be extended to the case with $n>2$ nodes and arbitrary $M\ge 1$: partition the $n$ nodes into two sets (say L and R), there must be at least $M$ patterns of communication between the two sets L and R. By applying this argument to all possible two-partitions such that $|L|=1$ and $|R|=n-1$, we can obtain a lower bound on the communication complexity as 
\begin{equation}
C_{AD(CD)}(n,M)\ge \frac{n}{2}\log_2 M.
\end{equation}
This lower bound is within a factor of $1/2$ of the upper bound we obtain previously, which implies that $C_{AD(CD)}(n,M) = \Theta(n\log_2 M)$. However, we can show that the lower bound of $\frac{n}{2}\log_2 M$ is generally not achievable. An example for this is the MEQ(3,4) problem. It can be shown that $C_{AD}(3,4) = C_{CD}(3,4) = 4$, while $\frac{n}{2}\log_2 M=3$. Details can be found in Appendix \ref{app:edge-coloring}.

The example above has demonstrated that, under our point-to-point communication model, we can no longer extend results from two-party communication complexity to multiparty version for tight bounds in the way it has been done under the broadcast communication models. The main reason for this is the lack of modeling of the ``networking'' aspect of the problem in both the two-party model and the broadcast communication models. In the two-party model, since there are only two nodes, no networking is necessary. In the broadcast communication models, all the nodes share a lot information from the broadcast and have roughly the same view of system, which makes it a not-so-distributed network. On the other hand, under our point-to-point communication modes, each node may only receive information from a subset of nodes; it is even possible that two nodes may receive information from two disjoint sets of nodes. As a result, different nodes can have very different views of the system. This makes the problem of finding the tight bound of communication complexity difficult, and new techniques may be required.

\section{Equivalent MEQ-AD Protocols}
\label{sec:eq_protocol}

In this section, we considers protocols $\in \Gamma_{AD}(n,M)$. A protocol $P$ is interpreted as a directed multi-graph $G(V,E)$, where the set of vertices $V=\{1,\cdots,n\}$ represents the $n$ nodes, and the set of directed edges $E=\{(T_1,R_1),\cdots,(T_{L(P)},R_{L(P)})\}$ represents the transmission schedule in each step. From now on, we will use the terms protocol and graph interchangeably, as well as the terms transmission and link.

We will say that a protocol $P$ is not better than protocol $P'$ if $C(P)\ge C(P')$. Two protocols $P$ and $P'$ are said to be {\bf equivalent} if they are not better than each other. The following lemma says that we can flip the direction of any link in a protocol $P$ and obtain a protocol that is equivalent to $P$.

\begin{lemma}\label{lm:flip_direction}
A protocol $P=\{\alpha_1\cdots\alpha_{l-1}\alpha_l\alpha_{l+1}\cdots\}$ with $\alpha_l=(T_l,R_l,f_l(x_{T_l},T_l^+))$ is equivalent to $P'=\{\alpha_1\cdots\alpha_{l-1}\alpha_l'\alpha_{l+1}'\cdots\}$ if the following conditions are all satisfied:
\begin{itemize}
\item $\alpha_l' = (R_l, T_l, f_l'(x_{R_l},R_l^+))$. Here $f_l'(x_{R_l},R_l^+)=f_l(x_{T_l},T_l^+)|_{x_1=\cdots=x_n=x_{R_l}}$ is the symbol that party $R_l$ {\em expects} to receive in step $l$ of protocol $P$, assuming all parties have the same input as $x_{R_l}$.

\item $\alpha_m' = (T_m,R_m,f_m'(x_{T_m},T_m^+))$ for $m>l$. 
\begin{itemize}
\item If $T_m= R_l$, $f_m'(x_{T_m},T_m^+)=f_m(x_{T_m},T_m^+)|_{f_l(x_{T_l},T_l^+)=f_l'(x_{R_l},R_l^+)}$  is the symbol that party $R_l$ sends in step $m$, {\bf pretending} that it has received $f_l'(x_{R_l},R_l^+)$ in step $l$ of $P$.
\item  If $T_m\neq R_l$, $f_m'(x_{T_m},T_m^+) = f_m(x_{T_m},T_m^+)$.
\end{itemize}

\item $T_l$ first computes $EQ_{T_l}$ in the same way as in $P$. Then $T_l$ sets $EQ_{T_l}=1$ if $f_l'(x_{R_l},R_l^+)\neq f_l(x_{T_l},T_l^+)$, else no change.
\end{itemize}
\end{lemma}
\begin{proof}
There are two cases:
\begin{itemize}
\item $f_l'(x_{R_l},R_l^+)=f_l(x_{T_l},T_l^+)$: It is not hard to see that in this case, the execution of every step is identical in both $P$ and $P'$, except for step $l$. So for all $i\neq T_l$, $EQ_i$ is identical in both protocols. Since $f_l'(x_{R_l},R_l^+)=f_l(x_{T_l},T_l^+)$,  $EQ_{T_l}$ remains unchanged, so it is also identical in both protocols.

\item $f_l'(x_{R_l},R_l^+)\neq f_l(x_{T_l},T_l^+)$: Observe that these two functions are different only if the inputs are not all identical. So it is correct to set $EQ_{T_l}=1$.
\end{itemize}
\end{proof}

Let us denote all the symbols a node $i$ receives from and sends to the other nodes throughout  the execution of protocol $P$ as $i^+$ and $i^-$, respectively. It is obvious that $i^-$ can be written as a function $F_i(x_i,i^+)$, which is the union of $f_l(x_i,i^+(l))$. If a protocol $P$ satisfies  $F_i(x_i,i^+) = F_i(x_i)$ for all $i$, we say $P$ is individual-input-determined (iid). The following lemma shows that there is always an iid equivalent for every protocol.

\begin{lemma}\label{lm:ordered}
For every protocol $P$, there always exists an iid equivalent protocol $P^*$, which corresponds to a  partially ordered acyclic graph.
\end{lemma}
\begin{proof} 
According to Lemma \ref{lm:flip_direction}, we can flip the direction of any edge in $G$ and obtain a new protocol which is equivalent to $P$. It is to be noted that we can keep flipping different edges in the graph, which implies that we can flip any subset of $E$ and obtain a new protocol equivalent to $P$. 

In particular, we consider a protocol equivalent to $P$, whose corresponding graph is acyclic, $i< j$ for all $(i,j)\in E$, with a proper re-labeling of the indexes. In this protocol, every node $i$ has no incoming links from any node with index greater than $i$. This implies that the symbols transmitted by node $i$ are independent of the nodes with larger indexes. Thus we can re-order the transmissions of this protocol such that node 1 transmits on all of its out-going links first, then node 2 transmits on all of its out-going links, ...,  node $n-1$ transmits to $n$ at the end. Name the new protocol $Q$. Obviously $Q$ is equivalent to $P$.

Since we can always find a protocol $Q$ equivalent to $P$ as described above, all we need to do now is to find $P^*$. If $Q$ itself is iid, then $P^*=Q$ and we are done. If not, we obtain $P^*$ in the following way (using function $F'$), which is similar to how we obtain the equivalent protocol $P'$ in Lemma \ref{lm:flip_direction}:
\begin{itemize}
\item For node 1, since it receives nothing from the other nodes, $F_1(x_1,1^+)=F_1(x_1)$ is trivially true.
\item For node $1<i<n$, we modify $Q$ as follows: node $i$ computes its out-going symbols as a function $F_i'(x_i) = F_i(x_i, i^+|_{x_1=\cdots = x_n=x_i})$, where $i^+|_{x_1=\cdots = x_n=x_i}$ is the expected incoming symbols to node $i$ assuming all parties have the same input as $i$. At the end, node $i$ checks if $i^+|_{x_1=\cdots = x_n=x_i}$ equals to the actual incoming symbols $i^+$. If they match, nothing is changed. If they do not match, the inputs can not be identical, and node $i$ can set $EQ_i = 1$. 
\end{itemize}
\end{proof}

Lemma \ref{lm:ordered} shows that, to study $C_{AD}(n,M)$, it is sufficient to investigate only protocols that are iid and partially ordered.

\section{MEQ-AD(3,6)}\label{sec:MEQ(3,6)}
Let us first consider MEQ-AD(3,6), i.e., the case where 3 nodes (say A, B and C) are trying to solve the MEQ-AD problem when each node is assigned one out of six values, namely $\{1,2,3,4,5,6\}$. According to Lemmas \ref{lm:flip_direction} and \ref{lm:ordered}, for any protocol that solves the MEQ-AD problem, there exists an equivalent iid partially ordered protocol in which node A has no incoming link, node B only transmits to node C, and node C has no out-going link. 
We construct one such protocol that solves MEQ-AD(3,6) and requires only 3 channel symbols, namely $\{1,2,3\}$, per link. The channel symbol being sent over link $ij$ is denoted as $s_{ij}$. Table \ref{tab:MEQ(3,6)} shows how $s_{ij}$ is computed as a function of $x_i$.

\begin{table}[h]
\centering
\begin{tabular}{|c||c|c|c|c|c|c|}
\hline
$x$ & 1 & 2 & 3 & 4 & 5 & 6\\
\hline
\hline
$s_{AB}$ & 1 & 1 & 2 & 2 & 3 & 3\\
\hline
$s_{AC}$ & 1 & 2 & 2 & 3 & 3 & 1\\
\hline
$s_{BC}$ & 1 & 2 & 3 & 1 & 2 & 3\\
\hline
\end{tabular}
\caption{A protocol for MEQ-AD(3,6)}
\label{tab:MEQ(3,6)}
\end{table}

For nodes B and C, they just compare the channel symbol received from each incoming link with the {\em expected} symbol computed with its own input value, and detect a mismatch if the received and expected symbols are not identical. For example, node B receives $s_{AB}(x_A)$ from node A. Then it detects a mismatch if the received $s_{AB}(x_A) \neq s_{AB}(x_B)$.

It can be easily verified that if the three input values are not all identical, at least one of nodes B and C will detect a mismatch. Hence the MEQ-AD(3,6) problem is solved with the proposed protocol. The communication complexity of this protocol is $3\log_2 3 = \log_2 27$. In fact, this protocol is optimal in the sense that it achieves the communication complexity of MEQ-AD(3,6):
\begin{equation}
C_{AD}(3,6) = \log_2 27.
\end{equation}
The proof of optimality can be found in Appendix \ref{app:edge-coloring}.

\section{MEQ-AD(3,$6^h$)}\label{sec:MEQ(3,$6^h$)}
Now consider the MEQ-AD(3,$6^h$) problem. One way to solve this problem is to extend the Table \ref{tab:MEQ(3,6)} as Table \ref{tab:MEQ(3,6^h)}. This protocol's communication complexity is $2\log_2 (6^h/2) + \log_2 3 = \log_2 (36^{h-1}\times 27)$.
\begin{table}[h]
\centering
\begin{tabular}{|c||c|c|c|c|c|c|c|c|c|c|c|}
\hline
$x$ & 1 & 2 & 3 & 4 & 5 & 6 & 7 & 8 & $\cdots$ & $6^h$-1 & $6^h$\\
\hline
\hline
$s_{AB}$ & 1 & 1 & 2 & 2 & 3 & 3 & 4 & 4 & $\cdots$ & $6^h/2$ & $6^h/2$\\
\hline
$s_{Ac}$ & 1 & 2 & 2 & 3 & 3 & 4 & 4 & 5 & $\cdots$ & $6^h/2$ & 1\\
\hline
$s_{Bc}$ & 1 & 2 & 3 & 1 & 2 & 3 & 1 & 2 & $\cdots$ & 2 & 3\\
\hline
\end{tabular}
\caption{A protocol for MEQ-AD(3,$6^h$)}
\label{tab:MEQ(3,6^h)}
\end{table}

 A more efficient way to solve this problem is to map each of the $6^h$ input values into a $h$-dimensional  vector in the vector space $\{1,2,3,4,5,6\}^h$, and then solve the MEQ-AD(3,6) problem $h$ times, one for each of the $h$ dimensions. Using the optimal MEQ-AD(3,6) protocol introduced in Section \ref{sec:MEQ(3,6)}, the communication complexity of this  protocol is $h\log_2 27= \log_2 27^h < \log_2 (36^{h-1}\times 27)$.

\section{MEQ-AD(3,$2^k$)}\label{sec:MEQ(3,$2^k$)}
Now we construct a protocol when the number of possible input values $M=2^k,k\ge 1$ and only binary symbols can be transmitted in each step, using the MEQ-AD(3,6) protocol we just introduced in the previous sections as a building block.

First, we map the $2^k$ input values into $2^k$ different vectors in the vector space $\{1,2,3,4,5,6\}^h$, where $h = \lceil k\log_6 2\rceil$. Then $h$ instances of the MEQ-AD(3,6) protocol are performed in parallel to compare the $h$ dimensions of the vector. Since 3 channels symbols are required for each instance of the MEQ-AD(3,6) protocol, we need to transmit a vector from $\{1,2,3\}^h$ on each of the links AB, AC and BC. One way to do so is to encode the $3^h$ possible vectors from $\{1,2,3\}^h$ into $b = \lceil h \log _2 3\rceil$ bits, and transmit the $b$ bits through the links. Since the $h$ instances of MEQ-AD(3,6) protocols solve the MEQ-AD(3,6) problem for each dimension, altogether they solve the MEQ-AD(3,$2^k$) problem.

The communication complexity of the proposed MEQ-AD(3,$2^k$) protocol can be easily computed as
\begin{eqnarray}
C(P) &=& 3\lceil h \log_2 3\rceil\\
&=& 3\lceil \lceil k\log_6 2\rceil \log_2 3\rceil \label{eq:ceiled}\\
&< & 3\lceil (k\log_6 2 +1) \log_2 3\rceil\\
&< & 3\left[ (k\log_6 2 +1) \log_2 3 +1\right]\\
&=& 3k\log_6 3 + 3(\log_2 3 + 1)\\
&<& 1.840k + 7.755 \label{eq:approx}
\end{eqnarray}

\begin{figure}[h]
\centering
\includegraphics[width = 4 in]{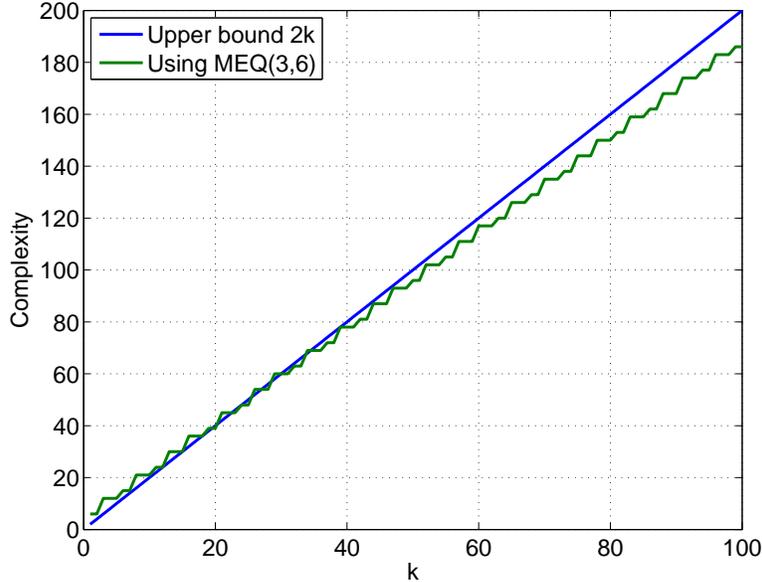}
\caption{Complexity of the proposed protocol v.s. upper bound 2k}
\label{fig:complexity}
\end{figure}
From Eq.\ref{eq:approx}, we can see that when $k$ is  large enough, the communication complexity this protocol becomes smaller than the upper bound $2\log_2 M = 2k$ from Section \ref{sec:upper}.  In Fig.\ref{fig:complexity}, we plot the communication complexity of this protocol, according to Eq.\ref{eq:ceiled}. It is easy to see that for $k>39$, $C(P) < 2k$.

The way in which the above protocol is constructed can be generalized to obtain a MEQ-AD(3,$M$) protocol $P$ with complexity 
\begin{equation}
C(P) < 1.840\log_2 M + \Delta\label{eq:approx_M}
\end{equation} 
for arbitrary value of $M$, where $\Delta$ is some positive constant.

\section{About MEQ-CD}\label{sec:MEQ-CD}
In this section, we will show that $C_{CD}(n,M)$ {\em roughly} equals to $C_{AD}(n,M)$:
\begin{equation}
C_{AD}(n,M)\le C_{CD}(n,M) \le C_{AD}(n,M)+n-2.
\end{equation}
We have shown the lower bound in Section \ref{sec:problem}. We will now prove the upper bound. 

Consider any partially ordered iid protocol $P\in \Gamma_{AD}(n,M)$ as described in Lemma \ref{lm:ordered}. We construct a protocol $P'$ by having node $i$ to send $EQ_i$ to node $n$ by the end of $P$, for all $1<i<n$. Node $n$ collects $n-1$ decisions (including $EQ_n$ as computed in $P$) from all nodes except for node 1. Then node $n$ compute the final decision
\begin{equation}
EQ'_n = \max\{EQ_2,\cdots,EQ_n\}.
\end{equation}
It is easy to see that, $EQ'_n = EQ(x_1,\cdots, x_n)$. So $P'\in \Gamma_{CD}(n,M)$. Since $C(P') = C(P) + n-2$, the upper bound is proved. From Eq.\ref{eq:approx_M} it then follows that there exist a protocol $P'$ that solves MEQ-CD(3,$M$) with complexity $C(P')\le 2\log_2 M$, for large enough $M$.

\section{Conclusion}
In this report, we study the  communication complexity problem of the multiparty equality function, under the point-to-point communication model. We demonstrate that traditional techniques generalized from two-party communication complexity problem are not sufficient to obtain tight bounds under the point-to-point communication model. We then introduce techniques to transform any MEQ-AD protocol into a equivalent partially ordered iid protocol. These techniques significantly reduce the space of MEQ-AD protocols to study. We then study the MEQ-AD(3,6) problem and introduce an optimal protocol that achieves $C_{AD}(3,6)$. This protocol is then used as building blocks for construction of efficient protocols for MEQ-AD(3,$6^h$) and MEQ-AD(3,$2^k$). The problem of finding the communication complexity of the MEQ problem for general values of $n$ and $M$ is still open.

\bibliographystyle{abbrv}
\bibliography{PaperList}

\appendix
\section{Edge Coloring Representation of MEQ-AD(3,$M$)}\label{app:edge-coloring}
From Sections \ref{sec:eq_protocol} and \ref{sec:MEQ(3,6)}, we have shown that it is sufficient to study 3-node systems where information is transmitted only on links AB, AC and BC. Let us denote $|s_{AB}|$, $|s_{AC}|$ and $|s_{BC}|$ as the number of different symbols being transmitted on links AB, AC and BC, respectively. Now consider the following simple bipartite graph $G(U,V,E)$, where $U$ and $V$ are the two disjoint sets of vertices and $E$ is the set of edges:
\begin{itemize}
\item $|U| = |s_{AB}|$, each vertex is labeled as $U_{s_{AB}(x)}$ for all $M$ values of $x$;
\item $|V| = |s_{AC}|$, each vertex is labeled as $V_{s_{AC}(x)}$ for all $M$ values of $x$;
\item $e_{ij} = (U_i,V_j) \in E$ if and only if $i=s_{AB}(x)$ and $j = s_{AC}(x)$ for some $x$.
\end{itemize}
In essence, each vertex $U_i$ (or $V_i$) represents the set of value $x$'s that produce the same value $s_{AB}(x)=i$ (or $s_{AC}(x) = i$); and each edge $e_{ij} = (U_i,V_j)$ represents the set of value $x$'s that produces the same pair of channel symbols $s_{AB}(x)=i$ and $s_{AC}(x)=j$. Let $|e_{ij}|$ be the size set of value $x$'s  corresponding to edge $e{ij}$. Fig.\ref{fig:bipartite} shows the bipartite graph corresponding to the MEQ-AD(3,6) protocol we introduced in Section \ref{sec:MEQ(3,6)}. Near the nodes $U_i$ and $V_i$ we show the set of value $x$'s such that $s_{AB}(x)=i$ and $s_{AC}(x)=i$, respectively. The number near each edges is the corresponding value of that edge.

\begin{figure}[t]
\centering
\includegraphics{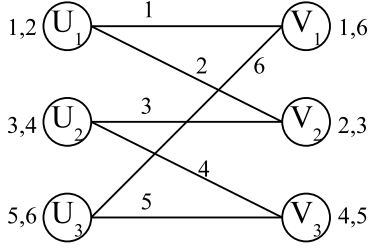}
\caption{The bipartite graph corresponding to the MEQ-AD(3,6) protocol from Section \ref{sec:MEQ(3,6)}.}
\label{fig:bipartite}
\end{figure}

We first argue that

\begin{lemma}\label{lm:bipartite}
$|e_{ij}| = 1$ for all $e_{ij}\in E$. Hence $|E|=M$ and $|U|\times|V|\ge M$.
\end{lemma}
\begin{proof}
Suppose to the contrary that there exists some $e_{ij}\in E$ with $|e_{ij}|\ge 2$. Then there must be two values $x,x'$ such that  $x\neq x'$, $s_{AB}(x)= s_{AB}(x')$ and $s_{AC}(x)= s_{AC}(x')$. Similar to the ``fooling set'' argument in Section \ref{sec:lower}, it is impossible for nodes B and C to tell the difference between the two input vectors $(x,x,x)$ and $(x',x,x)$, hence they will not be able to solve the MEQ-AD(3,$M$) problem, which leads to a contradiction. Then the first part of the lemma follows.

Since every edge represents one $x$, and there are $M$ possible values of $x$, it follows that $|E| = M$.
Also, in a simple bipartite graph, we always have $|U|\times|V|\ge |E|$. Thus $|U|\times|V|\ge M$.

\end{proof}

Since there is a one-to-one mapping from the set of input values $\{1,\cdots,M\}$ to the edges $E$, we will use the terms input value ($x$) and edge ($e_{ij}$) interchangeably. Now we prove the following theorem on the constraint of $s_{BC}$:

\begin{lemma}\label{lm:distance2}
$s_{BC}(x)\neq s_{BC}(x')$ if edges $x$ and $x'$ are adjacent or there is some other edge that is adjacent to both of them, in the bipartite graph $G(U,V,E)$.
\end{lemma}
\begin{proof}
Consider any pairs $x,x'$ such that $x\neq x'$ and $(U_{s_{AB}(x)},V_{s_{AC}(x')})\in E$. Let $x^*$ (maybe equal to $x$ or $x'$) be the input value that edge $(U_{s_{AB}(x)},V_{s_{AC}(x')})$ corresponds to. So we have $s_{AB}(x^*) = s_{AB}(x)$ and $s_{AC}(x^*) = s_{AC}(x')$. Now consider the input vector $(x^*,x,x')$  at node A, B and C. Since $s_{AB}(x^*) = s_{AB}(x)$, node B can not differentiate $(x^*,x,x')$ from $(x,x,x)$. So node B cannot detect the mismatch. Meanwhile, since $s_{AC}(x^*) = s_{AC}(x')$, node C can not differentiate $(x^*,x,x')$ from $(x',x',x')$ by just looking into the receives $s_{AC}$. So $s_{BC}(x)$ must be different from $s_{BC}(x')$, otherwise the MEQ-AD problem is not solved. Then the lemma follows. 
\end{proof}

Now we can conclude that the problem of designing $s_{BC}$, given functions $s_{AB}(\cdot)$ and $s_{AC}(\cdot)$, is equivalent to finding a distance-2 edge coloring for the corresponding bipartite graph $G(U,V,E)$. Furthermore, it should not be hard to see that any protocol $P$ that solves MEQ-AD(3,$M$) is equivalent to a bipartite graph $G(U,V,E)$ together with a distance-2 coloring scheme $W$ such that $|U|= |s_{AB}|$, $|V|=|s_{AC}|$, $|E|= M$, and $|W| = |s_{BC}|$, where $|W|$ denotes the number of colors in scheme $W$. Notice that
\begin{equation}
C(P) = \log_2|s_{AB}| + \log_2|s_{AC}| + \log_2|s_{BC}| 
= \log_2(|s_{AB}|\times |s_{AC}| \times |s_{BC}|).
\end{equation}
So we have the following theorem:
\begin{theorem}\label{thm:bipartite}
The existence of a MEQ-AD(3,$M$) protocol $P$ with complexity $C(P)$ is equivalent to the existence of a simple bipartite graph $G(U,V,E)$ together with a distance-2 coloring scheme $W$ such that $|U|\times |V| \times |W| = 2^{C(P)}$, given $|E| = M$, $|U|\times |V| \ge M$, $|U|\times |W| \ge M$ and $|V|\times |W|\ge M$.
\end{theorem}
\begin{proof}
The last two conditions come from the fact that we can flip the directions of the links and then apply Lemma \ref{lm:bipartite}.
\end{proof}

According to Theorem \ref{thm:bipartite}, we can conclude that the problem of finding $C_{AD}(3,M)$ is equivalent to the problem of finding the minimum of $|U|\times |V| \times |W|$ for the bipartite graphs and distance-2 coloring schemes that satisfy the above constrains.

Using Theorem \ref{thm:bipartite}, to show that $C_{AD}(3,4) = 4$, we only need to show that for every combination of $|U|\times |V| \times |W|<2^4=16$ there  exists no bipartite graph $G(U,V,E)$ and distance-2 coloring scheme $W$ that satisfy the conditions as described in Theorem \ref{thm:bipartite}. In other words, if the conditions are all satisfied, then the bipartite graph $G(U,V,E)$ cannot be distance-2 colored with $|W|$ colors. It is not hard to see that there are only two combinations (up to permutation) that satisfy all conditions and have product less than 16: (2, 2, 2) and (2, 2, 3). Notice that in both cases, $|E|=|U|\times|V|$, where every pair of edges are within distance of 2 of each other, which means graph $G$ can only be distance-2 colored with at least $|E|$ colors. Together with the upper bound from Section \ref{sec:upper}, this proves that $C_{AD}(3,4) = C_{CD}(3,4)= 4$.

Similarly, it can be shown that $C_{AD}(3,6) = \log_2 27$. There are only two combinations that satisfy all conditions in Theorem \ref{thm:bipartite} and have product less than 27: (2, 3, 3) and (2, 3, 4). Again, $|E| = |U|\times |V|$, so at least $|E|=6$ colors are needed, which proves that $C_{AD}(3,6) = \log_2 27$.

%
%

\end{document}